\documentclass[12pt, draftcls, journal, onecolumn]{IEEEtran}

\usepackage[pdftex]{graphicx}
\usepackage{float}
\usepackage{fixltx2e} 

\graphicspath{{./}}
\usepackage[cmex10]{amsmath}
\usepackage{epsfig,color}
\usepackage{mdwmath}
\usepackage{mdwtab}
\usepackage{amssymb}
\usepackage{subfig}
\usepackage{yhmath}
\usepackage{stackrel}
\usepackage{amsthm}

\newtheorem{theo}{Theorem}

\newtheorem{lemma}{Lemma}
\newtheorem{remark}{Remark}

\begin{document}
\title{On the Existence of an MVU Estimator for Target Localization with Censored, Noise Free Binary Detectors}
\author{Arian~Shoari,~\IEEEmembership{Student Member,~IEEE,}
        and Alireza~Seyedi,~\IEEEmembership{Senior Member,~IEEE,}
        \thanks{Our appreciation goes out to Mona Komeijani for preparing the illustrative graphs of this paper. Also we would like to thank Professor Mark Bocko and Professor Azadeh Vosoughi for their valuable feedback in finalizing this manuscript. A. Shoari is with the Department
of Electrical and Computer Engineering, University of Rochester, Rochester,
NY, E-mail: shoari@ece.rochester.edu. A. Seyedi sadly passed away in Oct 2014 before the final revision of this paper. He was with the Department of Electrical Engineering and Computer Science, University of Central Florida, Orlando, FL. E-mail: alireza.seyedi@ieee.org. }}



\maketitle

\begin{abstract}
The problem of target localization with censored noise free binary detectors is considered. In this setting only the detecting sensors report their locations to the fusion center. It is proven that if the radius of detection is not known to the fusion center, a minimum variance unbiased (MVU) estimator does not exist. Also it is shown that when the radius is known the center of mass of the possible target region is the MVU estimator. In addition, a sub-optimum estimator is introduced whose performance is close to the MVU estimator but is preferred computationally. Furthermore, minimal sufficient statistics have been provided, both when the detection radius is known and when it is not. Simulations confirmed that the derived MVU estimator outperforms several heuristic location estimators.
\end{abstract}


\newpage
\begin{tabular}{l l}
$\mathbb{G}$ &        Region of sensor deployment  \\
$\mathcal{A}(\mathbb{G})$ &     Hyper volume (area) of the region $\mathbb{G}$\\
$\mathbf{z}_{\text{T}}$ & Target location in the space\\
$\mathbf{O}$ & Origin \\
$N$ &  Total number of sensors deployed in the region\\
$n$ & Number of detecting sensors \\
$\mathcal{S}$ & Index of detecting sensors\\
$\mathcal{Z}$ & Set of locations of detecting sensors\\
$\mathbf{Z}$ &  Vector of locations of detecting sensors\\
$\mathbf{Z}_k$& Vector of locations of the first $k$ detecting sensors \\
$\mathcal{B}_{R}(\mathbf{z}_i)$ & A ball with radius $R$ around the $i$'th sensor\\
$\mathcal{T}(\mathbf{Z})$ & The \textit{possible target region} based on $\mathbf{Z}$ observation \\
$\mathcal{T}_j(\mathbf{Z})$ & The \textit{possible target region} with  sensor $j$th excluded from evaluation\\
$\mathcal{S}_{\text{MSS}}$ & Index of sensors forming the minimal sufficient statistic \\
$\mathcal{Z}_{\text{MSS}}$ & Set of locations of the sensors forming the minimal sufficient statistic \\
$\mathbf{Z}_{\text{MSS}}$ & Vector of the locations of sensors forming the minimal sufficient statistic\\
$\mathcal{MSS}(\mathbf{Z})$ & Function outputting $\mathcal{Z}_{\text{MSS}}$ from input $\mathbf{Z}$  \\
$MSS(\mathbf{Z})$  & Function outputting $\mathbf{Z}_{\text{MSS}}$ from input $\mathbf{Z}$  \\
$\mathbf{1}_X$ & The indicator function of $X$  \\
$\sigma(\mathcal{T})$ & Boundary surface of $\mathcal{T}$ \\
$f(\mathbf{Z};\mathbf{z}_{\text{T}})$ & Probability density function of occurrence of $\mathbf{Z}$ if the target located at $\boldsymbol{\mathbf{z}_{\text{T}}}$\\
$\mathcal{R}_{\Theta}(\mathbf{X})$ & Range of random variable $\mathbf{X}$ over parameter space $\Theta$\\
$\mathcal{C}(\mathcal{Z}_{\text{MSS}})$ & Convex hull of the set $\mathcal{Z}_{\text{MSS}}$ \\
$\mathcal{N}(\mathcal{T})$ & The set of points whose maximum distance from $\mathcal{T}$ are not more than  $R$ \\
$\mathbf{T_{k}}$ & A vector storing the distance of $(\mathbf{z}_1,..,\mathbf{z}_k)$ elements from $\mathbf{z}_1$ \\
$g_1\left(\mathbf{Z}\right)$ & An estimator based on observation $\mathbf{Z}$\\
$CM\left(\mathcal{T}(\mathbf{Z}, R_1)\right)$ & Center of mass of $\mathcal{T}(\mathbf{Z})$ with radius $R_1$ \\
\end{tabular}
\newpage

\section{Introduction}
Localization of an unknown transmitter with observations from a network of sensors is a well known problem in the literature \cite{Amundson09,Wang2012}. The observations can be carried out through measurement of Angle of Arrival (AoA) \cite{Oshman1999DOA,Kaplan2001DOA,Peng06AoA}, Time Difference of Arrival (TDoA) \cite{Kung1998TDOA,Yang2005TDOA,Yang2006TDOA2}, or Received Signal Strength (RSS) \cite{sheng2003collaborative,Sheng05,Blatt06,kieffer2006centralized,Vaghefi11_RSSbased,Rabbat2005,Ampeliotis08,Bulusu02,VarshneyFading07,VarshneyBinCoh,Arian2014SSP,ArianSpawc2010,Liu2007Target,Murthy2011multiple,Murthy2012multiple}. When the sensors are mobile, Frequency Difference of Arrival (FDoA) can also be used as an additional source of information \cite{Lee2007,Yu2012,Fuyong2014comments}. AoA, TDoA and FDoA approaches require sophisticated sensors, and, therefore do not fit well within the limitations of wireless sensor networks, specifically for energy and complexity constraints of the nodes. In practice, however, an exact(un-quantized) measurement is unrealistic because it requires unlimited bandwidth to communicate the data to the  fusion center. A binary RSS measurement is preferred because it is simpler and requires less resources. A number of papers describe methods that use binary RSS localization \cite{VarshneyBinCoh,Arian2014SSP,ArianSpawc2010,Liu2007Target,Murthy2011multiple,Murthy2012multiple}. In some papers it is assumed that the propagation model is isotropic and detection is noise free \cite{ArianSpawc2010,Liu2007Target,Murthy2011multiple,Murthy2012multiple}. This is equivalent to the situation when sensors average the power measurements over a long period of time, hence effectively eliminating the effect of noise \cite{Murthy2012multiple}. Moreover, the no noise regime provides a lower bound for location estimation in the presence of uncertainty such as noise or fading. Therefore, a minimum variance unbiased (MVU) estimator of this scenario will serve as a benchmark for comparison of  the performance of all target localizers. In addition, it provides insight into the effect of parameters such as density of sensor deployment and power in the localization process regardless of the environment and the type of detection. This assumption reduces the detection problem to the question whether the target is located within a certain radius of each sensor or not. However, the estimation problem is not well behaved, as the discontinuity in the proximity function violates the regularity conditions needed for obtaining Cram\'er-Rao  bound (CRB). In addition, some papers consider just the detecting sensors for evaluation of the target location \cite{ArianSpawc2010,artes2004target}. The motivation behind such approach is that when the region of sensor deployment is much larger than the detecting radius, the number of non-detecting sensor are much higher than the detecting ones. So censoring them will save a lot of communication and processing cost \cite{Rodriguez2004}, while knowing the location of detecting ones still provides a good estimate of the target location.\\
In this paper, we first study the redundancy in information provided by detecting sensors and find out minimal sufficient statistics for them. Then we investigate the existence of an MVU  estimator for this problem. Finally, we will introduce some sub-optimal estimators with low computational complexity that perform close to optimal and compare the performance of the MVU estimator and the sub-optimal estimator with some heuristic ones through simulation. To solve the problem in each stage, we have divided the problem into two separate cases: I) when the radius of detection is known which is equivalent to the situation when propagation model and the  transmit power are known to the fusion center; and II) when the radius of detection is unknown which is equivalent to the case when propagation model or transmit power are unknown such as the case in non-cooperative localization. \\
The rest of the paper is organized as follows : section \ref{sec:ProbFormulation} is the problem formulation, section \ref{sec:SufficientStatistics} study sufficient statistics for the observation, section \ref{sec:MVUexistance} investigate the existence of MVU estimator, section \ref{sec:suboptimum} discusses some sub-optimal estimator which behaves close to optimal, section \ref{sec:SimulationResults} reports the simulation results and section \ref{sec:Conclusion} is the conclusion.

\section{Problem Formulation}
\label{sec:ProbFormulation}
Assume that a target is located at unknown location $\mathbf{z}_{\text{T}}=[x_{1_{\text{T}}},x_{2_{\text{T}}},...,x_{l_{\text{T}}}]$ in $l$ dimensional space (in practical applications $l$ is either 2 or 3) and transmits a signal whose power propagates isotropically and is attenuated monotonically as a function of distance from the target. $N$ sensors, randomly scattered in a deployment region $\mathbb{G}$, with hyper volume $\mathcal{A}(\mathbb{G})$. They measure the received power and compare it with a threshold, $\tau$, to make a binary decision about the target presence. We assume the sensors do noise free decision, which can be considered as the  limiting case when the measured power is averaged over a sufficiently long duration. Furthermore, the sensors are configured such that only the detecting sensors report their locations, $\mathbf{z}_{1},..,\mathbf{z}_{n}$, to the fusion center where the localization is performed. Since the received power is a decreasing function of distance from the target, there is a ball around the target, $\mathcal{B}_{R}(\mathbf{z}_{\text{T}})$, where all the sensors inside will detect and those outside will not. From now on we call $R$ the detection radius. We assume that $\mathbb{G}$ is sufficiently large such that $\mathcal{B}_{2R}(\mathbf{z}_{\text{T}}) \subset \mathbb{G}$. In addition, we assume that at least one sensor detects the target. Let $n$ be the number of detecting sensors ($n \geq 1$) and $\mathcal{S}=\{1,..,n\}$ be the set of indices of all detecting sensors. Therefore, $\mathcal{Z}=\{\mathbf{z}_{i} | i \in \mathcal{S} \}$ will be the set of locations of all detecting sensors and $\mathbf{Z}=(\mathbf{z}_{i} | i \in \mathcal{S} )$ denote the vector contains those locations.

\section{Sufficient Statistics}
\label{sec:SufficientStatistics}
In the this section we derive sufficient statistics for estimation of target location $\mathbf{z}_{\text{T}}$. We consider the problem in two   cases depending on whether the detection radius, $R$, is known or unknown.

\subsection{Known Detection Radius}

Let us define the \textit{possible target region} given observation $\mathbf{Z}$  as \cite{Shenoy2005,Liu2004}
\begin{eqnarray}
\label{PCRDefinitionOne}
\mathcal{T}(\mathbf{Z})=\bigcap_{i \in \mathcal{S}} \mathcal{B}_{R}(\mathbf{z}_i), \nonumber
\end{eqnarray}
Alternatively since the order of $\mathbf{Z}$ elements does not matter in this definition, we may equivalently define $\mathcal{T}$ over the $\mathcal{Z}$ set i.e. \cite{Shenoy2005,Liu2004}
\begin{eqnarray}
\label{PCRDefinitionTwo}
\mathcal{T}(\mathcal{Z})=\bigcap_{i \in \mathcal{S}} \mathcal{B}_{R}(\mathbf{z}_i), \nonumber
\end{eqnarray}
Also let us define the \textit{possible target region} with node $j$ removed as
\begin{eqnarray}
\mathcal{T}_j(\mathbf{Z})= \bigcap_{i \in \mathcal{S}, i \neq j } \mathcal{B}_{R}(\mathbf{z}_i), \nonumber
\end{eqnarray}

Clearly, we have $\mathbf{z}_{\text{T}} \in \mathcal{T}(\mathbf{Z})$ and $\mathcal{T}(\mathbf{Z}) \subseteq \mathcal{T}_j(\mathbf{Z})$ for all $j$. Moreover, any sensor whose exclusion for evaluating $\mathcal{T}$ does not effect the result (i.e. $\mathcal{T}_j(\mathbf{Z})=\mathcal{T}(\mathbf{Z})$), does not provide any additional information about the target location.  Hence, by eliminating all such sensors, we can build a new set of indices, $\mathcal{S}_{\text{MSS}}$\footnote{The subscript MSS is used, since we will proceed to show that $\mathcal{Z}_{\text{MSS}}$ is in fact the minimum sufficient statistic.}, whose sensors do contribute in shaping the \textit{possible target region}:
\begin{eqnarray}
\mathcal{S}_{\text{MSS}}=\left\{ j | \mathcal{T}_j(\mathbf{Z}) \neq \mathcal{T}(\mathbf{Z}) \right\}, \nonumber
\end{eqnarray}
Correspondingly, let $\mathcal{Z}_{\text{MSS}}$ denote the set of sensors locations whose index are in $\mathcal{S}_{\text{MSS}}$:
\begin{eqnarray}
\mathcal{Z}_{\text{MSS}}=\left\{ \textbf{z}_j | j \in  \mathcal{S}_{\text{MSS}} \right\}, \nonumber
\end{eqnarray}
and let $\mathbf{Z}_{\text{MSS}}$ be the vector format of $\mathcal{Z}_{\text{MSS}}$ i.e. ,
\begin{eqnarray}
\mathbf{Z}_{\text{MSS}}=\left( \textbf{z}_j | j \in  \mathcal{S}_{\text{MSS}} \right), \nonumber
\end{eqnarray}
We may also represent $\mathcal{Z}_{\text{MSS}}$ and $\mathbf{Z}_{\text{MSS}}$ as a function of  $\mathbf{Z}$ i.e. $\mathcal{Z}_{\text{MSS}}=\mathcal{MSS}(\mathbf{Z})$ or $\mathbf{Z}_{\text{MSS}}=MSS(\mathbf{Z})$.

Recall that we have assumed that $\mathcal{S} \ne \varnothing$. Thus, so is $\mathcal{S}_{\text{MSS}}$, and the boundary surface of $\mathcal{T}$, denoted by $\sigma(\mathcal{T})$, is composed of hyper spherical domes centered at elements of $\mathcal{Z}_{\text{MSS}}$. 


Let $\mathfrak{Z}$ be a random variable vector having a probability density function $f(\mathbf{Z};\mathbf{z}_{\text{T}})$ and $\Theta$ be the $\mathbf{z}_{\text{T}}$ parameter space. Range of $\mathfrak{Z}$ over $\Theta$ would be defined as \cite{Mittelhammer2013}
\begin{eqnarray}
\mathcal{R}_{\Theta}(\mathfrak{Z})=\{\mathbf{Z}|\exists ~ \mathbf{z}_{\text{T}} \in \Theta; \,\, f(\mathbf{Z};\boldsymbol{\theta})>0 \}. \nonumber
\end{eqnarray}

In this paper, an observation $\mathbf{Z}$ is called a possible event if  $\mathbf{Z} \in \mathcal{R}_{\Theta}(\mathbf{\mathfrak{Z}})$.

\begin{lemma}\label{onetoone}
For any pair of possible events $\mathbf{Z}$ and $\mathbf{Z}'$,
\begin{eqnarray}
\mathcal{MSS}(\mathbf{Z})=\mathcal{MSS}(\mathbf{Z}')
\Leftrightarrow
\mathcal{T}(\mathbf{Z})=\mathcal{T}(\mathbf{Z}')
.  \nonumber
\end{eqnarray}
\end{lemma}
\begin{proof}
The forward direction is obvious from the construction of $\mathcal{MSS}(\mathbf{Z})$ and $\mathcal{MSS}(\mathbf{Z}')$.
For the backward direction let's assume that it is not true. Then,
\begin{eqnarray}
\exists X,Y \in \mathcal{R}_{\Theta}(\mathbf{\mathfrak{Z}}) ~,~ \mathcal{MSS}(X) \neq \mathcal{MSS}(Y) ~\&~\mathcal{T}(X)=\mathcal{T}(Y) \nonumber
\end{eqnarray}
Thus, there is at least one uncommon element between $\mathcal{MSS}(X)$ and $ \mathcal{MSS}(Y)$. Without loss of generality we assume it to be $\mathbf{z_p} \in \mathcal{MSS}(X)$ and $\mathbf{z_p} \notin  \mathcal{MSS}(Y)$. Let $\sigma$ represents the surface of a set. $\sigma(\mathcal{B}_{R,\mathbf{z_p}}) \cap \sigma (\mathcal{T}(X))$ is a segment of a hyper sphere dome and $l+1$ points can be selected on it such that they are not located in the same hyper plane. Since $\mathcal{T}(X)=\mathcal{T}(Y)$, $\sigma(\mathcal{B}_{R,\mathbf{z_p}}) \cap \sigma(\mathcal{T}(Y))$ is also a segment of the same dome and contains those $l+1$ points. But any $l+1$ points on a dome will determine exactly one sensor location belonging to $\mathcal{MSS}(Y)$ which has $R$ distance from them all. Thus, $\mathbf{z_p} \in \mathcal{MSS}(Y)$ which contradicts our assumption and the proof is complete.
\end{proof}

\begin{theo}
\label{theoMinimalSuffSt}
$\mathcal{Z}_{\text{MSS}}$ is a minimal sufficient statistic for the estimation of target location.
\end{theo}

\begin{proof}
Due to Lemma \ref{onetoone}, it suffices to show that $\mathcal{T}(\mathbf{Z})$ is a minimal sufficient statistics for the estimation of target location.
The probability density function that the detecting sensor will be located at, $\mathbf{Z}$, can be described as
\begin{eqnarray}
\label{fZmain}
f(\mathbf{Z};\mathbf{z}_{\text{T}}) = \left\{ \begin{array}{lcl}
 \left(\frac{1}{\mathcal{A}(\mathbb{G})}  \right)^{n(\mathbf{Z})}    \left(\frac{\mathcal{A}(\mathbb{G})-\mathcal{A}(\mathcal{B}_{R}(\mathbf{\mathbf{z}_{\text{T}}}))    }{\mathcal{A}(\mathbb{G})}  \right)^{N-n(\mathbf{Z})}    &
\mathbf{z}_{\text{T}} \in \mathcal{T}(\mathbf{Z}) \\
0  &  \mathbf{z}_{\text{T}} \notin \mathcal{T}(\mathbf{Z})  \\
\end{array}\right.,
\end{eqnarray}
where $n(\mathbf{Z})$ represents the number of elements of vector $\mathbf{Z}$ and $\mathcal{A}(.)$ represents the hyper volume of the shape (or simply the area in two dimensional space).  Equation (\ref{fZmain}) can be rewritten as following which permits the Neyman-Fisher factorization
\begin{eqnarray}
f(\mathbf{Z};\mathbf{z}_{\text{T}}) &=& \left(\frac{1}{\mathcal{A}(\mathbb{G})}\right)^{n(\mathbf{Z})} \left( \frac{\mathcal{A}(\mathbb{G})-\mathcal{A}(\mathcal{B}_{R}(\mathbf{z}_{\text{T}}))}{\mathcal{A}(\mathbb{G})}\right)^{N-n(\mathbf{Z})} \mathbf{1}_{\mathbf{z}_{\text{T}} \in \mathcal{T}(\mathbf{Z})} \nonumber\\
&=& h(\mathbf{Z}) g(\mathbf{z}_{\text{T}},\mathcal{T}(\mathbf{Z})), \nonumber
\end{eqnarray}
where $\mathbf{1}_X$ is the indicator function of $X$, $h(\mathbf{Z})=\left(\frac{1}{\mathcal{A}(\mathbb{G})}\right)^{n(\mathbf{Z})} \left( \frac{\mathcal{A}(\mathbb{G})-\mathcal{A}(\mathcal{B}_{R}(\mathbf{z}_{\text{T}}))}{\mathcal{A}(\mathbb{G})}\right)^{N-n(\mathbf{Z})}$ and $g(\mathbf{z}_{\text{T}},\mathcal{T}(\mathbf{Z}))=\mathbf{1}_{\mathbf{z}_{\text{T}} \in \mathcal{T}(\mathbf{Z})}$. Thus, $\mathcal{T}(\mathbf{Z})$ is a sufficient statistics for estimation of $\mathbf{z}_{\text{T}}$. \\

To prove that it is minimal we consider another sufficient statistics $\mathcal{U}(.,.)$ and show that $\mathcal{T}$ is a function of $\mathcal{U}$ through showing that
\begin{eqnarray}
\forall ~~ \mathbf{Z}_1, \mathbf{Z}_2 \in \mathcal{R}_{\Theta}(\mathfrak{Z}), ~~ \mathcal{U}(\mathbf{Z}_1)=\mathcal{U}(\mathbf{Z}_2) \Rightarrow  \mathcal{T}(\mathbf{Z}_1)=\mathcal{T}(\mathbf{Z}_2) \label{eq:minimalU}
\end{eqnarray}
where $\mathfrak{Z}$ represents the random variable vector for observations and $\mathbf{Z}_1$ and $\mathbf{Z}_2$ are two instances of that random variable vector. Assume that \eqref{eq:minimalU} does not hold, then
\begin{eqnarray}
\exists ~~ \mathbf{Z}_1, \mathbf{Z}_2 \in \mathcal{R}_{\Theta}(\mathfrak{Z});  \mathcal{T}(\mathbf{Z}_1) \neq \mathcal{T}(\mathbf{Z}_2) \mbox{ and } \mathcal{U}(\mathbf{Z}_1)=\mathcal{U}(\mathbf{Z}_2). eq:minimalU
\end{eqnarray}
Assume that
$\mathbf{z}_{\text{T}} \in  \left((\mathcal{T}(\mathbf{Z}_1)-\mathcal{T}(\mathbf{Z}_2)) \cup (\mathcal{T}(\mathbf{Z}_2) - \mathcal{T}(\mathbf{Z}_1))   \right)$. Without loss of generality we assume that $\mathbf{z}_{\text{T}} \in \mathcal{T}(\mathbf{Z}_1)-\mathcal{T}(\mathbf{Z}_2)$. Then $f(\mathbf{Z}_1;\mathbf{z}_{\text{T}})\neq0$ and $ f(\mathbf{Z}_2;\mathbf{z}_{\text{T}})=0$. \\
However, due to Fisher Neyman factorization we have
\begin{eqnarray}
f(\mathbf{Z}_1;\mathbf{z}_{\text{T}}) & = & h_{\mathcal{U}}(\mathbf{Z}_1) g_{\mathcal{U}}(\mathbf{z}_{\text{T}},\mathcal{U}(\mathbf{Z}_1)) \neq 0\nonumber \\
f(\mathbf{Z}_2;\mathbf{z}_{\text{T}}) & = & h_{\mathcal{U}}(\mathbf{Z}_2) g_{\mathcal{U}}(\mathbf{z}_{\text{T}},\mathcal{U}(\mathbf{Z}_2))=h_{\mathcal{U}}(\mathbf{Z}_2) g_{\mathcal{U}}(\mathbf{z}_{\text{T}},\mathcal{U}(\mathbf{Z}_1)) = 0\nonumber
\end{eqnarray}
Since the first equation is non zero, therefore \footnote{sub-script $\mathcal{U}$ is used to differentiate the factorization for sufficient statistics $\mathcal{U}$ from the one done for sufficient statistics $\mathcal{T}$} $g_{\mathcal{U}}(\mathbf{z}_{\text{T}},\mathcal{U}(\mathbf{Z}_1)) \neq 0$. Hence $h_{\mathcal{U}}(\mathbf{Z}_2)=0$, which means that $\forall \mathbf{z}_{\text{T}} , f(\mathbf{Z}_2;\mathbf{z}_{\text{T}})=0 $ i.e. $\mathbf{Z}_2 \notin \mathcal{R}_{\Theta}(\mathfrak{Z})$ (due to the definition of $\mathcal{R}_{\Theta}(\mathfrak{Z})$) which contradicts our assumption that $\mathbf{Z}_2 \in \mathcal{R}_{\Theta}(\mathfrak{Z})$ and the proof is complete.
\end{proof}

\begin{remark}
\label{rem:notcomplete1}
$\mathcal{Z}_{\text{MSS}}$ is not complete. To show this, we provide a measurable function $\mathbf{g}(.)\ne \mathbf{0}$ where $E[\mathbf{g}(\mathcal{Z}_{\text{MSS}})]=\mathbf{0}$ for all $\mathbf{z}_{\text{T}}$. Define
the centroid (average) of $\mathcal{Z}_{\text{MSS}}$ by $\mathbf{\bar{z}}=\frac{1}{|\mathcal{Z}_{\text{MSS}}|} \sum_{i \in \mathcal{S}_{\text{MSS}}} \mathbf{z}_{i}$, and the center of mass of convex hull of $\mathcal{Z}_{\text{MSS}}$ by $\mathbf{\tilde{z}}=\frac{1}{\int_{\mathcal{C}(\mathcal{Z}_{\text{MSS}})} dv} \int_{\mathcal{C}(\mathcal{Z}_{\text{MSS}})} \mathbf{z} dv$ where $\mathcal{C}(\mathcal{Z}_{\text{MSS}})$ represent convex hull of the set, $\mathcal{Z}_{\text{MSS}}$, and $\int \mathbf{z} dv$ is the integral over hyper volume. Since $\mathbb{G}$ is assumed to be sufficiently large, the distribution of the location of detecting sensors is isotropic around the target. Thus, so is the distribution of $\mathbf{\bar{z}}$ and $\mathbf{\tilde{z}}$. Thus $E[\mathbf{\bar{z}}]=E[\mathbf{\tilde{z}}]=\mathbf{z}_{\text{T}}$. In general, when dimension of the space is higher than 1, $\mathbf{\bar{z}} \neq \mathbf{\tilde{z}}$. Consequently, for $\mathbf{g}(\mathcal{Z}_{\text{MSS}})=\mathbf{\bar{z}}-\mathbf{\tilde{z}}$ we have $E[\mathbf{g}(\mathcal{Z}_{\text{MSS}})]=E[\mathbf{\bar{z}}-\mathbf{\tilde{z}}]=\mathbf{0}$ for all $\mathbf{z}_{\text{T}}$ but $\mathbf{\bar{z}}-\mathbf{\tilde{z}} \neq \mathbf{0}$.\\
In one dimensional space, the $\mathcal{Z}_{\text{MSS}}$ simply becomes the maximum and minimum of $\mathbf{z}_i \in \mathcal{Z}$. An equivalent problem to this case has been studied in \cite{David2003OrderStatistics} and the conclusion is that when $R$ is known, $\mathcal{Z}_{\text{MSS}}$ is not complete.
\end{remark}

%
%

\begin{theo}
\label{theoconvex}
If $\mathcal{Z} \neq \emptyset $, then $\mathcal{T}(\mathbf{Z})$ is convex.
\end{theo}

\begin{proof} We draft the proof by mathematical induction. Define $\mathbf{Z}_k=(\mathbf{z}_1,..,\mathbf{z}_k)$ where $k < n$. Because $\mathcal{Z} \neq \emptyset$, $\mathcal{T}(\mathbf{Z}_1)$ is a ball and hence is convex. If $\mathcal{T}(\mathbf{Z}_k)$ is convex, so does   $\mathcal{T}(\mathbf{Z}_{k+1})=\mathcal{T}(\mathbf{Z}_k) \cap \mathcal{B}_{R}(\mathbf{z}_{k+1})$ because both $\mathcal{T}(\mathbf{Z}_k)$ and $\mathcal{B}_{R}(\mathbf{z}_{k+1})$ are convex.
\end{proof}
%

\subsection{Unknown Detection Radius}
In this sub section we assume that $R$ is unknown. This implies that the target's transmit power is not known. Let the polytope $\mathcal{C}(\mathcal{Z})=\{ \sum_i \lambda_i \mathbf{z}_i| \mathbf{z}_i \in \mathcal{Z}, \sum_i \lambda_i=1, \lambda_i \geq0  \}$ be the convex hull of the locations of the detecting sensors \cite{boyd2004convex}. Moreover, let $\mathcal{W}=\{\mathbf{w}_i\} \subset \mathcal{Z} $ be the set of locations of corners of $\mathcal{C}(\mathcal{Z})$.

\begin{lemma}\label{lemma:D}
Any sensors located in $\mathcal{C}(\mathcal{Z})$ is a detecting sensor.
\end{lemma}
\begin{proof}
Since $\mathbf{z}_i \in \mathcal{C}(\mathcal{Z})$, it can be represented as $\mathbf{z}_i= \sum_j \lambda_j \mathbf{w}_j $ where $\sum_j \lambda_j=1 $ and $\lambda_j \geq 0$. Thus,
\begin{eqnarray}
\| \mathbf{z}_i-\mathbf{z}_{\text{T}}\| =\left\|\sum_j \lambda_j \mathbf{w}_j-\mathbf{z}_{\text{T}} \sum_j \lambda_j\right\|=\left\|\sum_j \lambda_j (\mathbf{w}_j-\mathbf{z}_{\text{T}})\right\|  \leq \sum_j \lambda_j  \| \mathbf{w}_j-\mathbf{z}_{\text{T}} \|  \leq \sum_j \lambda_j R=R. \nonumber
\end{eqnarray}
\end{proof}

\begin{theo}\label{SSunknownR}
$\mathcal{W}$ is a sufficient statistic for estimating $\mathbf{z}_{\text{T}}$ and $R$.
\end{theo}
\begin{proof}
The probability density function that the detecting sensors will be located at $\mathbf{Z}$ is
\begin{eqnarray}
f\left(\mathbf{Z}; \mathbf{z}_{\text{T}} , R \right) &=& \left(\frac{1}{\mathcal{A}(\mathbb{G})}  \right)^{n(\mathbf{Z})}    \left(\frac{\mathcal{A}(\mathbb{G})-\mathcal{A}(\mathcal{B}_{R}(\mathbf{\mathbf{z}_{\text{T}}}))    }{\mathcal{A}(\mathbb{G})}  \right)^{N-n(\mathbf{Z})}   \prod_{i \in \mathcal{S}}  \mathbf{1}_{\mathbf{z}_i \in   \mathcal{B}_{R}(\mathbf{z}_{\text{T}}) }  \nonumber\\
&=& \left(\frac{1}{\mathcal{A}(\mathbb{G})}  \right)^{n(\mathbf{Z})}    \left(\frac{\mathcal{A}(\mathbb{G})-\mathcal{A}(\mathcal{B}_{R}(\mathbf{O}))    }{\mathcal{A}(\mathbb{G})}  \right)^{N-n(\mathbf{Z})}   \prod_{i \in \mathcal{S}}  \mathbf{1}_{\mathbf{z}_i \in   \mathcal{B}_{R}(\mathbf{z}_{\text{T}}) }  \nonumber
\label{f_ZD_gen}
\end{eqnarray}
where $n(\mathbf{Z})$ represents the number of elements of vector $\mathbf{Z}$ and $\mathbf{O}$ is the origin. Due to Lemma \ref{lemma:D}, if all elements of $\mathcal{W}$ are inside $\mathcal{B}_{R}(\mathbf{z}_{\text{T}})$, so is $\mathcal{C}(\mathcal{Z})$. Thus all detecting sensors are also inside $\mathcal{B}_{R}(\mathbf{z}_{\text{T}})$ and the value of corresponding indicator functions are 1. On the other hand, if any $\mathbf{z}_l \in \mathcal{W}$ is not located inside $\mathcal{B}_{R}(\mathbf{z}_{\text{T}})$, $f\left(\mathbf{Z}; \mathbf{z}_{\text{T}} , R \right)$ would be zero regardless of value of indicator function for other sensors. Hence, equation \eqref{f_ZD_gen} becomes

\begin{eqnarray}
f\left(\mathbf{Z}; \mathbf{z}_{\text{T}} , R \right)
&=& \left(\frac{1}{\mathcal{A}(\mathbb{G})}  \right)^{n(\mathbf{Z})}    \left(\frac{\mathcal{A}(\mathbb{G})-\mathcal{A}(\mathcal{B}_{R}(\mathbf{O}))  }{\mathcal{A}(\mathbb{G})}  \right)^{N-n(\mathbf{Z})} \prod_{z_l \in \mathcal{W}} \mathbf{1}_{\mathbf{z}_l \in   \mathcal{B}_{R}(\mathbf{z}_{\text{T}}) }
\label{PDFZ-D}
\end{eqnarray}
Fisher-Neyman factorization of (\ref{PDFZ-D}) is
\begin{eqnarray}
f(\mathbf{Z}; \mathbf{z}_{\text{T}} , R)  =h(\mathbf{Z}) g_{\mathbf{z}_{\text{T}},R} \left(\mathcal{W} \right), \nonumber
\label{f_ZD_SuffSt}
\end{eqnarray}
where
\begin{align}
\left\{ \begin{array}{lcl}
g_{\mathbf{z}_{\text{T}},R}\left(\mathcal{W} \right)  & = & \prod_{z_l \in \mathcal{W}} \mathbf{1}_{\mathbf{z}_l \in   \mathcal{B}_{R}(\mathbf{z}_{\text{T}}) }, \\
h(\mathbf{Z})  & = & \left(\frac{1}{\mathcal{A}(\mathbb{G})}  \right)^{n(\mathbf{Z})}    \left(\frac{\mathcal{A}(\mathbb{G})-\mathcal{A}(\mathcal{B}_{R}(\mathbf{O}))  }{\mathcal{A}(\mathbb{G})}  \right)^{N-n(\mathbf{Z})}  .
\end{array}\right.
\end{align}


This implies that, $\mathcal{W}$ is a sufficient statistic for estimation of $\mathbf{z}_{\text{T}}$ and $R$ due to Fisher-Neyman factorization theorem \cite{Kay93Estimationbook}.
\end{proof}

It's worth noting that for some large $R$, all corners of $\mathcal{W}$ contribute to the \textit{possible target region}. We could not come up with a formal proof for this in $l$ dimensional space but a conjecture in two dimensional space is that if it is not, the non-contributing corner would be located in line with the other two adjacent corners which contradicts that it is a corner of convex hull. Thus, with the same reasoning as proof of theorem \ref{theoMinimalSuffSt}, $\mathcal{W}$ would be a minimal sufficient statistics. Note that in this case $R,\mathbf{z}_{\text{T}}$ are both parameters of estimation.

\begin{remark}\label{rem:notcomplete}
Although $\mathcal{W}$ is a sufficient statistic for estimation of target location, it is not complete. To see this we will provide a measurable function $\mathbf{g}(.)\ne \mathbf{0}$ where $E[\mathbf{g}(\mathcal{W})]=\mathbf{0}$ for all $\mathbf{z}_{\text{T}}$. Denote the centroid (average) of $\mathcal{W}$ by $\mathbf{\bar{z}}(\mathcal{W})=\frac{1}{|\mathcal{W}|}\sum_{j \in \mathcal{W}} \mathbf{\mathbf{w}_j}$, and the center of mass of $\mathcal{C}(\mathcal{W})$ by $\mathbf{\tilde{z}}(\mathcal{W})=\frac{1}{\int_{\mathcal{C}(\mathcal{W})}~dv} \int_{\mathcal{C}(\mathcal{W})} \mathbf{z} ~dv$.  As in Remark \ref{rem:notcomplete1}, it is straight forward to show that $E[\mathbf{\bar{z}}]=E[\mathbf{\tilde{z}}]=\mathbf{z}_{\text{T}}$. On the other hand, in general $\mathbf{\bar{z}} \neq \mathbf{\tilde{z}}$ (except in a one dimensional space). Consequently, for $\mathbf{g}(\mathcal{W})=\mathbf{\bar{z}}(\mathcal{W})-\mathbf{\tilde{z}}(\mathcal{W})$ we have $E[\mathbf{g}(\mathcal{W})]=E[\mathbf{\bar{z}}-\mathbf{\tilde{z}}]=\mathbf{0}$ for all $\mathbf{z}_{\text{T}}$ but $\mathbf{g}(\mathcal{W}) \neq \mathbf{0}$.\\
The one dimensional space is an exception in this regard and has been studied in \cite{David2003OrderStatistics}. The conclusion is that when $R$ is unknown, $\mathcal{Z}_{\text{MSS}}$ is complete in this case.
\end{remark}

\begin{remark}
Recall that  in Theorem \ref{SSunknownR} we showed  that $\mathcal{W}$ is a sufficient statistic when $R$ is unknown. Hence, this set is also a sufficient statistic for the case when $R$ is known. That is, if we define
\begin{eqnarray}
\mathcal{S}_{\text{SS}} & = & \left\{ i| \mathbf{z}_i \in \mathcal{W}   \right\} \nonumber\\
\mathcal{Z}_{\text{SS}}& = &\left\{ \mathbf{z}_i | i \in  \mathcal{S}_{\text{SS}} \right\} .\nonumber
\end{eqnarray}
we have $\mathcal{Z}_{\text{MSS}}\subseteq \mathcal{Z}_{\text{SS}}$. Thus, considering that finding the convex hull of detecting sensors is a straight forward procedure, $\mathcal{S}_{\text{SS}}$ can be used as a starting point for building $\mathcal{S}_{\text{MSS}}$ for the case when $R$ is known.
\begin{eqnarray}
\mathcal{S}_{\text{MSS}} & = & \left\{ k | k \in \mathcal{S}_{\text{SS}}  ,   \mathcal{T}_k( \mathcal{Z}_{\text{SS}}) \neq \mathcal{T}(\mathcal{Z}_{\text{SS}}) \right\}.  \nonumber
\end{eqnarray}

\end{remark}

\section{On the Existence of MVU Estimator}
\label{sec:MVUexistance}
It is easy to show that the likelihood function is discontinuous, and thus, the regularity conditions required to obtain Cram\'er-Rao  bound (CRB) does not hold. On the other hand, since the minimum sufficient statistics is not complete, Lehmann-Scheff\'e theorem \cite{Kay93Estimationbook} cannot be used to find the MVU estimator. In this chapter we explore strategies to find the MVU estimator directly from its definition while CRB can not be established. The approach lead to interesting results regarding the existence of an MVU estimator in the cases when $R$ is known or when it is not known.

\subsection{Known Detection Radius}

\begin{lemma}
Assuming a target location $\mathbf{z}_{\text{T}}$, an observation $\mathbf{Z}$ is a \textit{possible observation} i.e. $f(\mathbf{Z};\mathbf{z}_{\text{T}})>0$ if and only if  $\mathbf{z}_{\text{T}} \in \mathcal{T}(\mathbf{Z}) $
\end{lemma}

\begin{proof} The forward direction has already been shown during definition of $\mathcal{T}$. The reverse is true because
\begin{eqnarray}
&& \mathbf{z}_{\text{T}} \in \mathcal{T}(\mathbf{Z}) \Rightarrow \mathbf{z}_{\text{T}} \in \bigcap_{i \in \mathcal{S}} \mathcal{B}_{R,\mathbf{z_i}} \nonumber\\
&& \Rightarrow  \forall i \in \mathcal{S}, \mbox{ } \| \mathbf{z}_{\text{T}}- \mathbf{z_i} \| \leq R  \nonumber
\end{eqnarray}

Thus that observation is a \textit{possible observation}.
\end{proof}

Let us define \textit{near points}, $\mathcal{N}(\mathcal{T}(\mathbf{Z}_{\text{mss}}))$ as all the points whose maximum distance from $\mathcal{T}(\mathbf{Z}_{\text{mss}})$ are less than $R$, i.e.,\\
\begin{eqnarray}
\mathcal{N}(\mathcal{T}(\mathbf{Z}_{\text{mss}}))=\{\mathbf{z} | \mathbf{z} \in \mathbb{G} \& \max\limits_{\mathbf{z}_{H} \in \mathcal{T}(\mathbf{Z}_{\text{mss}})} \| \mathbf{z}-\mathbf{z}_{H} \| \leq R \} \nonumber
\end{eqnarray}
in other words if a detecting sensor located inside $\mathcal{N}(\mathcal{T}(\mathbf{Z}_{\text{mss}}))$, it does not contribute to $\mathcal{T}$ and is not a member of minimal sufficient statistics.

Figure \ref{IlustrativePCR_N_F}(a) shows an example of such region. The detecting sensors are demonstrated by small empty circles. The target is shown by an empty triangle. The white area is $\mathcal{T}(\mathbf{Z}_{\text{mss}})$ region, dark gray area illustrates the boundaries of $\mathcal{N}(\mathcal{T}(\mathbf{Z}_{\text{mss}}))$. The black circle identifies the detection disk. \\
Note that if $\mathcal{T}(\mathbf{Z}_{\text{mss}}) \neq \emptyset$ , $\mathcal{N}(\mathcal{T}(\mathbf{Z}_{\text{mss}}))$ is not empty because at least $\mathbf{Z}_{\text{mss}} \subset \mathcal{N}(\mathcal{T}(\mathbf{Z}_{\text{mss}}))$.

\begin{lemma}
\label{NinB_R}
$\mathcal{N}(\mathcal{T}(\mathbf{Z}_{\text{mss}})) ~ \subset~ \mathcal{B}_R(\mathbf{z}_{\text{T}})$.
\begin{proof}
We know that $\mathbf{z}_{\text{T}} \in \mathcal{T}(\mathbf{Z}_{\text{mss}})$. Thus, according to definition of $\mathcal{N}$ we will have,
\begin{eqnarray}
\forall ~ \mathbf{z}_{P} \in \mathcal{N}(\mathcal{T}(\mathbf{Z}_{\text{mss}})), ~ \| \mathbf{z}_{P}- \mathbf{z}_{\text{T}} \| \leq R  ~\Rightarrow  \mathbf{z}_{P} \in \mathcal{B}_{R}(\mathbf{z}_{\text{T}})  \nonumber
\end{eqnarray}
thus, the proof is complete.
\end{proof}
\end{lemma}

\begin{figure*}[!t]
\begin{tabular}{cc}
\includegraphics[width=3.2 in]{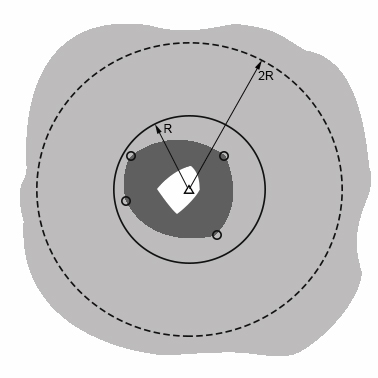}
  &
  \includegraphics[width=3.2 in]{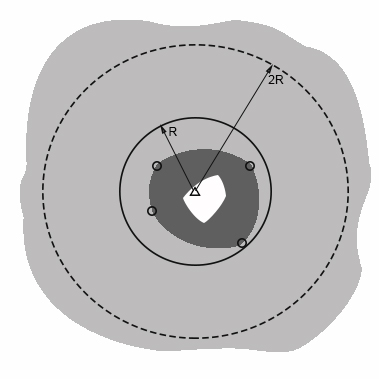} \\
  (a)&(b)\\
\end{tabular}
\caption{$\mathcal{T}$ and $\mathcal{N}(\mathcal{T})$ shifts along with $\mathbf{Z}_{\text{mss}}$}
\label{IlustrativePCR_N_F}
\end{figure*}

\begin{figure}[!t]
\begin{centering}
\includegraphics[width=3.2 in]{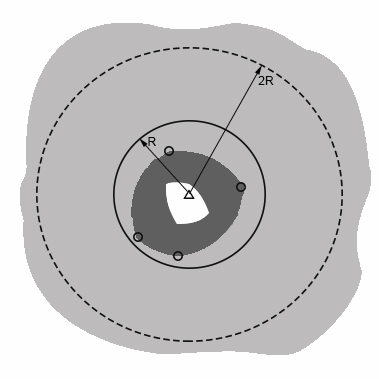}
\caption{$\mathcal{T}$ and $\mathcal{N}(\mathcal{T})$ rotates along with $\mathbf{Z}_{\text{mss}}$}
\label{RotatePCR}
\end{centering}
\end{figure}

\begin{theo}
\label{centerofmass}
Center of mass of $\mathcal{T}(\mathbf{Z})$ is the minimum variance estimator conditioned that at least one sensor is detecting and the target is located well inside the sensor deployment region such that $\mathcal{B}_{2R}(\mathbf{z}_{\text{T}}) \subset \mathbb{G}$.
\end{theo}

\begin{proof} Assume that $\mathcal{Z}_{\text{mss}}$ is the minimal sufficient statistics of $\mathbf{Z}$ and $\mathbf{Z}_{\text{mss}}$ is the vector format of that.  Then according to lemma \ref{onetoone}, $\mathcal{T}(\mathbf{Z})=\mathcal{T}(\mathbf{Z}_{\text{MSS}})$. \\
Let's assume that an MVU estimator exist for estimation of $\mathbf{z}_{\text{T}}$ and is represented by $\theta(\mathbf{Z})$. Since $\mathcal{Z}_{\text{mss}}$ is the minimal sufficient statistics, $h(\mathcal{Z}_{\text{mss}})=E[\theta(\mathbf{Z})\mid \mathcal{Z}_{\text{mss}}]$ should also be an MVU estimator according to Rao-Blackwell theorem \cite{Kay93Estimationbook}. Moreover, because $\mathcal{Z}_{\text{mss}}$ can be considered as a function of $\mathbf{Z}_{\text{mss}}$, we can restrict our search for MVU estimator to functions of $\mathbf{Z}_{\text{mss}}$.
In other words, if MVU estimator exists we can express it as a function of $\mathbf{Z}_{\text{mss}}$ i.e.
\begin{eqnarray}
\hat{\mathbf{z}_{\text{T}}}=g\left(\mathbf{Z}_{\text{mss}}\right)
\end{eqnarray}
On the other hand, the conditions $\mathcal{B}_{2R}(\mathbf{z}_{\text{T}}) \subset \mathbb{G}$ and at least one sensor would be detecting guarantees that for any possible $\mathcal{T}=\mathcal{T}(\mathbf{Z}_{\text{mss}})$, if $\mathbf{Z}_{\text{mss}}$ elements shift by some vector $\mathbf{d}$, $|\mathbf{d}|<R$ then $\mathcal{N}(\mathcal{T}(\mathbf{Z}_{\text{mss}}))$ will not reach to the borders of $\mathbb{G}$ and its area remains fix. We will calculate the variance of location estimation by conditioning over the number of detecting sensors, $n$,  and over the number of sensors forming the minimal sufficient statistics denoted by $k$. That's because the incident of each ($n$,$k$) partitions the random space to disjoint subspace.\\
Let us represents  $(\mathbf{z}_{1},..,\mathbf{z}_{k})$ by $\mathbf{Z_{k}}$. Considering that below expectation is symmetrical regardless of the order of the sensors forming the minimal sufficient statistics, we assume that the sensors belong to minimum sufficient statistics have indexes $1,..,k$ and generalize the result by including a binomial coefficient. The assumption that the minimal sufficient statistics are located at index $1,..,k$ dictates that other detecting sensors should be located in $\mathcal{N}(\mathcal{T}(\mathbf{Z_{k}}))$. Therefore, the expectation of square error of the estimation conditioned that at least one sensor is detecting can be calculated as following:
\begin{eqnarray}
E[\left|\hat{\mathbf{z}_{\text{T}}}-\mathbf{z}_{\text{T}}\right|^2]
&=& \sum_{n=1}^{N} \left(\frac{\mathcal{A}(\mathbb{G})-\mathcal{A}(\mathcal{B}_{R}(\mathbf{O}))}{\mathcal{A}(\mathbb{G})} \right)^{N-n} \binom{N}{n} \sum_{k=1}^{n} \binom{n}{k}   \nonumber\\
& &  \underbrace{\int_{\mathcal{B}_R(\mathbf{z}_{\text{T}})} .. \int_{\mathcal{B}_R(\mathbf{z}_{\text{T}})}}_{k}  \underbrace{\int_{\mathcal{N}(\mathcal{T}(\mathbf{Z_{k}}))} .. \int_{\mathcal{N}(\mathcal{T}(\mathbf{Z_{k}}))}}_{n-k}   \nonumber\\
& & \left(\frac{1}{\mathcal{A}(\mathbb{G})}\right)^{n}
\left| g\left(\mathbf{Z_{k}}\right) -\mathbf{z}_{\text{T}}\right|^2 \mathbf{1}_{\mathbf{Z_{k}}=MSS(\mathbf{Z_{k}})}  \nonumber\\
& & d \mathbf{z}_{n}..d \mathbf{z}_{k+1} d \mathbf{z}_{k}..d \mathbf{z}_{1}     \nonumber
\end{eqnarray}

where $N$ represents the total number of sensors, $\int_{\mathbb{Q} } d \mathbf{z}=\underbrace{\int..\int_{\mathbb{Q}}}_{l} d x_1 ..d x_l$ is the integral over the multi dimensional volume, $\mathbb{Q}$, and $\left(\frac{\mathcal{A}(\mathbb{G})-\mathcal{A}(\mathcal{B}_{R}(\mathbf{O}))}{\mathcal{A}(\mathbb{G})} \right)^{N-n}$ is the probability that $N-n$ sensors located outside the detecting disk. The indicator function $\mathbf{1}_{\mathbf{Z_{k}}=MSS(\mathbf{Z_{k}})}$ make the integrand zero whenever $\mathbf{Z_{k}}$ is not the minimal sufficient statistics.
\begin{eqnarray}
E[\left|\hat{\mathbf{z}_{\text{T}}}-\mathbf{z}_{\text{T}}\right|^2]
&=& \sum_{n=1}^{N} \binom{N}{n}  \sum_{k=1}^{n} \binom{n}{k}  \frac{\left( \mathcal{A}(\mathbb{G})-\mathcal{A}(\mathcal{B}_{R}(\mathbf{O}))\right)^{N-n}}{\mathcal{A}(\mathbb{G})^N}     \nonumber\\
& & \underbrace{\int_{\mathcal{B}_R(\mathbf{z}_{\text{T}})} .. \int_{\mathcal{B}_R(\mathbf{z}_{\text{T}})}}_{k}  \underbrace{\int_{\mathcal{N}(\mathcal{T}(\mathbf{Z_{k}}))} .. \int_{\mathcal{N}(\mathcal{T}(\mathbf{Z_{k}}))}}_{n-k} d \mathbf{z}_{n}..d \mathbf{z}_{k+1} \nonumber\\
& & \mathbf{1}_{\mathbf{Z_{k}}=MSS(\mathbf{Z_{k}})}  \left| g\left(\mathbf{Z_{k}}\right) -\mathbf{z}_{\text{T}}\right|^2  d \mathbf{z}_{k}..d \mathbf{z}_{1}     \nonumber\\
&=& \sum_{n=1}^{N}  \binom{N}{n} \sum_{k=1}^{n} \binom{n}{k} \frac{\left( \mathcal{A}(\mathbb{G})-\mathcal{A}(\mathcal{B}_{R}(\mathbf{O}))\right)^{N-n}}{\mathcal{A}(\mathbb{G})^N}   \underbrace{\int_{\mathcal{B}_R(\mathbf{z}_{\text{T}})} .. \int_{\mathcal{B}_R(\mathbf{z}_{\text{T}})}}_{k}   \nonumber\\
& & \left(\mathcal{A}\left(\mathcal{N}(\mathcal{T}(\mathbf{Z_{k}}))\right)\right)^{n-k}
\left| g\left(\mathbf{Z_{k}}\right) -\mathbf{z}_{\text{T}}\right|^2 \mathbf{1}_{\mathbf{Z_{k}}=MSS(\mathbf{Z_{k}})}  \nonumber\\
& &  d \mathbf{z}_{k}..d \mathbf{z}_{1}     \nonumber
\end{eqnarray}
Now, if we change the variables according to the following rules
\begin{eqnarray}
\mathbf{t}_1&=&0  \nonumber\\
\mathbf{t}_2&=&\mathbf{z}_2-\mathbf{z}_1  \nonumber\\
...\nonumber\\
\mathbf{t}_k&=&\mathbf{z}_k-\mathbf{z}_1  \nonumber
\end{eqnarray}
and assuming $\mathbf{T_{k}}=(\mathbf{0},\mathbf{t}_2..,\mathbf{t}_k)$ and $\mathbf{z}_1 \mathbf{1_k}=\underbrace{(\mathbf{z}_1,..,\mathbf{z}_1)}_{k}$, then $(\mathbf{z}_1,..,\mathbf{z}_k)=\mathbf{z}_1 \mathbf{1_k}+\mathbf{T_k}$. Moreover, it is reasonable to assume that the MVU estimator, $g$, shifts in space whenever the the input data shifts i.e. $g\left(\mathbf{z}_1\mathbf{1_k}+\mathbf{T_k}\right)=g\left(\mathbf{T_{k}}\right)+\mathbf{z}_1$
cause otherwise the $g$ would be dependent on the selection of the origin.
\begin{eqnarray}
E[\left|\hat{\mathbf{z}_{\text{T}}}-\mathbf{z}_{\text{T}}\right|^2]
&=& \sum_{n=1}^{N} \binom{N}{n}  \sum_{k=1}^{n} \binom{n}{k}   \frac{\left( \mathcal{A}(\mathbb{G})-\mathcal{A}(\mathcal{B}_{R}(\mathbf{O}))\right)^{N-n}}{\mathcal{A}(\mathbb{G})^N}   \nonumber\\
& & \int_{\mathcal{B}_R(\mathbf{z}_{\text{T}})} \underbrace{\int_{\mathbb{S}(\mathbf{z}_1)} .. \int_{\mathbb{S}(\mathbf{z}_1)}}_{k-1}
 \left(\mathcal{A}\left(\mathcal{N}(\mathcal{T}(\mathbf{z}_1 \mathbf{1_{k}}+\mathbf{T_{k}}))\right)\right)^{n-k}
\mathbf{1}_{\mathbf{z}_{\text{T}} \in \mathcal{T}(\mathbf{z}_1 \mathbf{1_{k}}+\mathbf{T_{k}})}    \times \nonumber\\
& &   \left| g\left(\mathbf{T_{k}}\right) +\mathbf{z}_1-\mathbf{z}_{\text{T}}\right|^2 \mathbf{1}_{(\mathbf{z}_1 \mathbf{1_{k}}+\mathbf{T_{k}})=MSS(\mathbf{z}_1 \mathbf{1_{k}}+\mathbf{T_{k}})}   d \mathbf{t}_{k}..d \mathbf{t}_{2} d \mathbf{z}_{1}    \nonumber
\end{eqnarray}
where the interval of integrals for elements of $\mathbf{T_{k}}$ are the shifted versions of $\mathcal{B}_{R}(\mathbf{z}_{\text{T}})$ by $\mathbf{z}_1$  i.e.  $\mathbb{S}(\mathbf{z}_1)=\mathcal{B}_{R}(\mathbf{z}_{\text{T}})-\mathbf{z}_1$. Moreover, the indicator function, $\mathbf{1}_{\mathbf{z}_{\text{T}} \in \mathcal{T}(\mathbf{z}_1 \mathbf{1_{k}}+\mathbf{T_{k}})}$,  make the integrand zero whenever $\mathbf{T_{k}}$ is an impossible (invalid) combination. \\
The next equation uses the fact that when
$\mathcal{B}_{2R}(\mathbf{z}_{\text{T}}) \subset \mathbb{G}$, shifting $\mathbf{T_{k}}$ by $\mathbf{z}_1$ will not change the area of $\mathcal{A}(\mathcal{N}(\mathcal{T}))$ i.e.
$\mathcal{A}(\mathcal{N}(\mathcal{T}(\mathbf{T_{k}}+\mathbf{z}_1 \mathbf{1_{k}})))=\mathcal{A}(\mathcal{N}(\mathcal{T}(\mathbf{T_{k}})))$ as illustrated in figures \ref{IlustrativePCR_N_F}(a) and \ref{IlustrativePCR_N_F}(b). That's because according to lemma (\ref{NinB_R}), $\mathcal{N}(\mathcal{T}(\mathbf{T_{k}}+\mathbf{z}_1 \mathbf{1_{k}})) \subset \mathcal{B}_{R}(\mathbf{z}_{\text{T}}) $. It is also worth noting that the result will not change if we substitute the inner integral intervals, $\mathbb{S}(\mathbf{z}_1)$, with $\mathcal{B}_{2R}(O)= \mathcal{B}_{R}(O) \oplus \mathcal{B}_{R}(O)$ (where $\oplus$ represents the Minkowski sum \cite{oks2006minkowski}) because $\mathbb{S}(\mathbf{z}_1) \subset \mathcal{B}_{2R}(O)$ and the indicator functions, $\mathbf{1}$, guarantees to make the integrands zero whenever $\mathbf{t_i} \notin \mathbb{S}(\mathbf{z}_1) \& \mathbf{t_i} \in \mathcal{B}_{2R}(O)$ for any $i \in \{2,..,k\}$. Thus, we will have\\
\begin{eqnarray}
E[\left|\hat{\mathbf{z}_{\text{T}}}-\mathbf{z}_{\text{T}}\right|^2]
&=& \sum_{n=1}^{N}  \binom{N}{n}  \sum_{k=1}^{n} \binom{n}{k}   \frac{\left( \mathcal{A}(\mathbb{G})-\mathcal{A}(\mathcal{B}_{R}(\mathbf{O}))\right)^{N-n}}{\mathcal{A}(\mathbb{G})^N}   \nonumber\\
& & \int_{\mathcal{B}_R(\mathbf{z}_{\text{T}})} \underbrace{\int_{\mathcal{B}_{2R}(\mathbf{O})} .. \int_{\mathcal{B}_{2R}(\mathbf{O})}}_{k-1}
\left(\mathcal{A}\left(\mathcal{N}(\mathcal{T}(\mathbf{T_{k}}))\right)\right)^{n-k}
\mathbf{1}_{\mathbf{z}_{\text{T}} \in \mathcal{T}(\mathbf{z}_1 \mathbf{1_{k}}+\mathbf{T_{k}})}    \times \nonumber\\
& &
 \left| g\left(\mathbf{T_{k}}\right) +\mathbf{z}_1-\mathbf{z}_{\text{T}}\right|^2  \mathbf{1}_{(\mathbf{z}_1 \mathbf{1_{k}}+\mathbf{T_{k}})=MSS(\mathbf{z}_1 \mathbf{1_{k}}+\mathbf{T_{k}})}    d \mathbf{t}_{k}..d \mathbf{t}_{2} d \mathbf{z}_{1}    \nonumber
\end{eqnarray}
On the other hand, it is obvious from the mechanism of building $\mathcal{T}$ that $\mathcal{T}(\mathbf{z}_1 \mathbf{1_{k}}+\mathbf{T_{k}})$  would be the shifted version of $\mathcal{T}(\mathbf{T_{k}})$ by $\mathbf{z}_1$, and similarly $MSS(\mathbf{z}_1 \mathbf{1_{k}}+\mathbf{T_{k}})$ would be the shifted version of $MSS(\mathbf{T_{k}})$  by $\mathbf{z}_1$. Thus,
\begin{eqnarray}
E[\left|\hat{\mathbf{z}_{\text{T}}}-\mathbf{z}_{\text{T}}\right|^2]
&=& \sum_{n=1}^{N}  \binom{N}{n}  \sum_{k=1}^{n} \binom{n}{k}   \frac{\left( \mathcal{A}(\mathbb{G})-\mathcal{A}(\mathcal{B}_{R}(\mathbf{O}))\right)^{N-n}}{\mathcal{A}(\mathbb{G})^N}   \nonumber\\
& & \int_{\mathcal{B}_R(\mathbf{z}_{\text{T}})} \underbrace{\int_{\mathcal{B}_{2R}(\mathbf{O})} .. \int_{\mathcal{B}_{2R}(\mathbf{O})}}_{k-1}
\left(\mathcal{A}\left(\mathcal{N}(\mathcal{T}(\mathbf{T_{k}}))\right)\right)^{n-k} \mathbf{1}_{\mathbf{z}_{\text{T}}- \mathbf{z}_1 \in \mathcal{T}(\mathbf{T_{k}})}  \times \nonumber\\
& &  \left| g\left(\mathbf{T_{k}}\right) +\mathbf{z}_1-\mathbf{z}_{\text{T}}\right|^2    \mathbf{1}_{\mathbf{T_{k}}=MSS(\mathbf{T_{k}})}  d \mathbf{t}_{k}..d \mathbf{t}_{2} d \mathbf{z}_{1}      \nonumber
\end{eqnarray}
Now with a change of variable $\mathbf{t'}=\mathbf{z}_{\text{T}}-\mathbf{z}_1$, we have
\begin{eqnarray}
E[\left|\hat{\mathbf{z}_{\text{T}}}-\mathbf{z}_{\text{T}}\right|^2]
&=& \sum_{n=1}^{N}  \binom{N}{n} \sum_{k=1}^{n} \binom{n}{k}  \frac{\left( \mathcal{A}(\mathbb{G})-\mathcal{A}(\mathcal{B}_{R}(\mathbf{O}))\right)^{N-n}}{\mathcal{A}(\mathbb{G})^N}    \nonumber\\
& & \int_{\mathbf{z}_{\text{T}}-\mathcal{B}_R(\mathbf{z}_{\text{T}})} \underbrace{\int_{\mathcal{B}_{2R}(\mathbf{O})} .. \int_{\mathcal{B}_{2R}(\mathbf{O})}}_{k-1}
\left(\mathcal{A}\left(\mathcal{N}(\mathcal{T}(\mathbf{T_{k}}))\right)\right)^{n-k} \mathbf{1}_{\mathbf{T_{k}}=MSS(\mathbf{T_{k}})}    \times \nonumber\\
& & \left| g\left(\mathbf{T_{k}}\right) -\mathbf{t'}\right|^2  \mathbf{1}_{\mathbf{t'} \in \mathcal{T}(\mathbf{T_{k}})}     d \mathbf{t}_{k}..d \mathbf{t}_{2} d \mathbf{t'}      \nonumber
\end{eqnarray}
Now because the intervals of the inner integrals are independent of $\mathbf{t'}$, we can change the order of outer integral with the inner ones. Moreover, because the first integral interval is taken over the hyper volume, we know that $\int_{\mathbf{z}_{\text{T}}-\mathcal{B}_R(\mathbf{z}_{\text{T}})} .. d \mathbf{t'}=\int_{\mathcal{B}_R(\mathbf{z}_{\text{T}})-\mathbf{z}_{\text{T}}}.. d(- \mathbf{t'})=\int_{\mathcal{B}_R(\mathbf{O})}.. d(-\mathbf{t'})= \int_{\mathcal{B}_R(\mathbf{O})} .. d(\mathbf{t'})$, thus we can reorganize the integral as:
 \begin{eqnarray}
E[\left|\hat{\mathbf{z}_{\text{T}}}-\mathbf{z}_{\text{T}}\right|^2]
&=& \sum_{n=1}^{N}  \binom{N}{n}  \sum_{k=1}^{n} \binom{n}{k} \frac{\left( \mathcal{A}(\mathbb{G})-\mathcal{A}(\mathcal{B}_{R}(\mathbf{O}))\right)^{N-n}}{\mathcal{A}(\mathbb{G})^N}   \nonumber\\
& & \underbrace{\int_{\mathcal{B}_{2R}(\mathbf{O})} .. \int_{\mathcal{B}_{2R}(\mathbf{O})}}_{k-1}
 \left(\mathcal{A}\left(\mathcal{N}(\mathcal{T}(\mathbf{T_{k}}))\right)\right)^{n-k}
\mathbf{1}_{\mathbf{T_{k}}=MSS(\mathbf{T_{k}})}  \times \nonumber\\
& &
\left( \int_{\mathcal{B}_R(\mathbf{O})}   \mathbf{1}_{\mathbf{t'} \in \mathcal{T}(\mathbf{T_{k}})}    \left| g\left(\mathbf{T_{k}}\right) -\mathbf{t'}\right|^2  d \mathbf{t'}  \right)  d \mathbf{t}_{k}..d \mathbf{t}_{2}      \nonumber
\end{eqnarray}
because at least one sensor is assumed to be detecting $\mathcal{T}(\mathbf{T_{k}}) \subset \mathcal{B}_{R}(O)$. Therefore, the last integral can be written over $\mathcal{T}(\mathbf{T_{k}})$ i.e.,
\begin{eqnarray}
E[\left|\hat{\mathbf{z}_{\text{T}}}-\mathbf{z}_{\text{T}}\right|^2]
&=& \sum_{n=1}^{N} \binom{N}{n} \sum_{k=1}^{n} \binom{n}{k} \frac{\left( \mathcal{A}(\mathbb{G})-\mathcal{A}(\mathcal{B}_{R}(\mathbf{O}))\right)^{N-n}}{\mathcal{A}(\mathbb{G})^N}   \nonumber\\
& &  \underbrace{\int_{\mathcal{B}_{2R}(\mathbf{O})} .. \int_{\mathcal{B}_{2R}(\mathbf{O})}}_{k-1}
 \left(\mathcal{A}\left(\mathcal{N}(\mathcal{T}(\mathbf{T_{k}}))\right)\right)^{n-k}
\mathbf{1}_{\mathbf{T_{k}}=MSS(\mathbf{T_{k}})}  \times \nonumber\\
& &
\left( \int_{\mathcal{T}(\mathbf{T_{k}})}   \left| g\left(\mathbf{T_{k}}\right) -\mathbf{t'}\right|^2  d \mathbf{t'}  \right)  d \mathbf{t}_{k}..d \mathbf{t}_{2}
\end{eqnarray}
which is indicating that to make the expectation minimum, $g\left(\mathbf{T_{k}}\right)$ should be the center of mass of $\mathcal{T}(\mathbf{T_{k}})$, i.e.
\begin{eqnarray}
\label{gofT}
g\left(\mathbf{T_{k}}\right)=CM\left(\mathcal{T}(\mathbf{T_{k}})\right)
\end{eqnarray}
where $CM$ denotes the center of mass. Now that the answer to the minimization problem is known, it is easy to verify that  $g\left(\mathbf{T_{k}}\right)=CM\left(\mathcal{T}(\mathbf{T_{k}})\right)
$ makes the variance minimum because $g\left(\mathbf{T_{k}}\right)$ is independent from $\mathbf{t'}$ and the only point in space which has minimum average distance from all $\mathbf{t'} \in \mathcal{T}(\mathbf{T_{k}})$ is $CM\left(\mathcal{T}(\mathbf{T_{k}})\right)$. In other words, it makes the inner integral minimum and $g\left(\mathbf{T_{k}}\right)$  does not show up in any other parts of the expression. It's worth mentioning here that although center of mass of a \textit{possible target region} has been used as a localizer in literature \cite{He2003,Lazos2005,Lazos2004}, it appears that it has never been proved to be the MVU estimator.    \\
Furthermore, similar to previous derivation we can show that $g\left(\mathbf{T_{k}}\right)=CM\left(\mathcal{T}(\mathbf{T_{k}})\right)$ is unbiased as following:
\begin{eqnarray}
E[\hat{\mathbf{z}_{\text{T}}}-\mathbf{z}_{\text{T}}]
&=& \sum_{n=1}^{N} \binom{N}{n} \sum_{k=1}^{n} \binom{n}{k} \frac{\left( \mathcal{A}(\mathbb{G})-\mathcal{A}(\mathcal{B}_{R}(\mathbf{O}))\right)^{N-n}}{\mathcal{A}(\mathbb{G})^N}   \nonumber\\
& &  \underbrace{\int_{\mathcal{B}_{2R}(\mathbf{O})} .. \int_{\mathcal{B}_{2R}(\mathbf{O})}}_{k-1}
\left(\mathcal{A}\left(\mathcal{N}(\mathcal{T}(\mathbf{T_{k}}))\right)\right)^{n-k}
\mathbf{1}_{\mathbf{T_{k}}=MSS(\mathbf{T_{k}})}  \times \nonumber\\
& &
\left( \int_{\mathcal{T}(\mathbf{T_{k}})}   \mathbf{1}_{\mathbf{t'} \in \mathcal{T}(\mathbf{T_{k}})}    \left( g\left(\mathbf{T_{k}}\right) -\mathbf{t'}\right)  d \mathbf{t'}  \right)  d \mathbf{t}_{k}..d \mathbf{t}_{2}      \nonumber\\
&=& \sum_{n=1}^{N}  \binom{N}{n} \sum_{k=1}^{n} \binom{n}{k}  \frac{\left( \mathcal{A}(\mathbb{G})-\mathcal{A}(\mathcal{B}_{R}(\mathbf{O}))\right)^{N-n}}{\mathcal{A}(\mathbb{G})^N}  \nonumber\\
& & \underbrace{\int_{\mathcal{B}_{2R}(\mathbf{O})} .. \int_{\mathcal{B}_{2R}(\mathbf{O})}}_{k-1}
 \left(\mathcal{A}\left(\mathcal{N}(\mathcal{T}(\mathbf{T_{k}}))\right)\right)^{n-k}
\mathbf{1}_{\mathbf{T_{k}}=MSS(\mathbf{T_{k}})}  \times \nonumber\\
& &
\left(  0 \right) d \mathbf{t}_{k}..d \mathbf{t}_{2}      \nonumber\\
&=& 0
\end{eqnarray}

Thus, MVU estimator exists and is unique. From (\ref{gofT}) and shifting sensor locations by $\mathbf{z}_1$, we conclude that for any $\mathbf{Z} \in \mathcal{R}_{\Theta}(\mathfrak{Z})$,
\begin{eqnarray}
\label{gofZ}
g\left(\mathbf{Z}\right)=CM\left(\mathcal{T}(\mathbf{Z})\right)
\end{eqnarray}
is the minimum variance unbiased estimator for $\mathbf{z}_{\text{T}}$.
\end{proof}

\subsection{Unknown Detection Radius}
For this problem the estimation parameters are $[\mathbf{z}_{\text{T}},R]$ and we know that $\mathcal{Z}_{\text{SS}}$ is a sufficient statistics. We consider two detecting range $R_1$ and $R_2$. Let us define $\mathcal{T}(\mathbf{Z},R)=\bigcap_{i \in \mathcal{S}} \mathcal{B}_{R}(\mathbf{z}_i)$. Based on previous section analysis, $g_1\left(\mathbf{Z}\right)=CM\left(\mathcal{T}(\mathbf{Z},R_1)\right)$ is the unique unbiased minimum variance estimator for $\mathbf{z}_{\text{T}}$ when $R=R_1$ and $g_2\left(\mathbf{T}\right)=CM\left(\ mathcal{T}(\mathbf{Z},R_2)\right)$ is the unique unbiased minimum variance estimator for $\mathbf{z}_{\text{T}}$ when $R=R_2$. Note that no matter what the $R$ parameter is, $g_2$ and $g_1$ remains  unbiased for location estimation because the probability density function of observation is isotropic and if the observation rotates, so do $g_2$ and $g_1$. These two functions can be different for a specific realization of $\mathbf{Z}$ as demonstrated in an example in figure \ref{PCRR1_1vsPCRR2_0_9}. Therefore, the MVU estimator does not exist for this case because for different parameter, $R$, the minimum variance unbiased estimators are different.
\begin{figure}[!t]
\begin{centering}
\includegraphics[width=3.5 in]{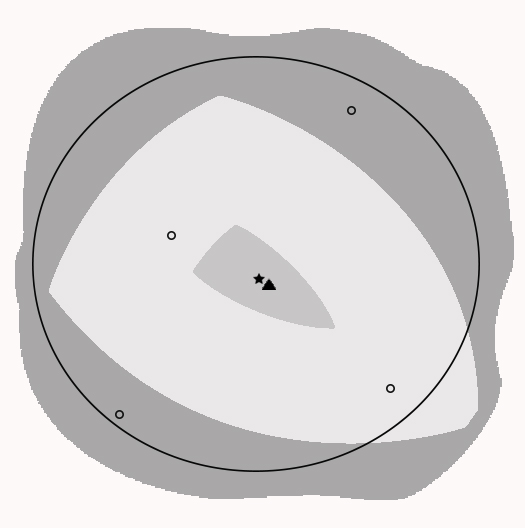}
\caption{Illustration of $\mathcal{T}(\mathbf{Z},R)$ and corresponding $CM$  for $R=1.05$ and $R=1.6$}
\label{PCRR1_1vsPCRR2_0_9}
\end{centering}
\end{figure}

\section{Sub-Optimum Estimators in Two Dimensional Space}
\label{sec:suboptimum}

\subsection{When Detecting Radius is Known}
Although we proved in theorem \ref{theoconvex} that $\mathcal{T}(\mathbf{Z})$ is a convex shape, it is not an easy shape to work with. Not only its visualization but also finding the center of mass of that shape is computationally costly. 
An intuitive approach to decrease the complexity is to find the center of mass of the convex hull of the corners of $\mathcal{T}(\mathbf{Z})$ as illustrated in figure \ref{IllustrationSuboptimal}.\\ 
The center of mass of this shape can be found easily in two dimensional space by triangulating it through a triangulation algorithm such as the ones discovered by Euler, Fournier or Toussaint \cite{Boissonnat1989}-\cite{Guibas1987}.
Assume the result of triangulation to be $q$ triangles each represented by its corners $A_i$,$B_i$,$C_i$ as $\Delta_{A_i,B_i,C_i}$, then
\begin{eqnarray}
CM\left(\mathcal{SLT}(\mathbf{Z})\right)=\sum_{i=1}^{q} \frac{\mathcal{A}(\Delta_{A_i,B_i,C_i})}{\sum_{i=1}^{q} \mathcal{A}(\Delta_{A_i,B_i,C_i})} \times \frac{A_i+B_i+C_i}{3}
\end{eqnarray}
would be a sub-optimum estimator for $\mathbf{z}_{\text{T}}$ where $\mathcal{SLT}(\mathbf{Z})$ is the convex polygon generated by connecting the consecutive corners of $\mathcal{T}(\mathbf{Z})$ by straight lines, and $\frac{A_i+B_i+C_i}{3}$ is the center of mass for the $i$th triangle


\subsection{When Detecting Radius is Unknown}
As we discussed in previous section, the MVU estimator does not exist for this case. Still the MVU of the case when $R$ is known act as a Clairvoyant\footnote{A Clairvoyant estimator is referred to an estimator with advanced knowledge of some parameters of estimation as if it is provided by a genie.} estimator and its performance act as the lower bound for all estimators who lacks the knowledge of $R$.  We may notice that all estimators provided in \cite{ArianSpawc2010} can be employed when $R$ is unknown. In the next section we will see the simulation results and we will find out that center of MEC performance follows the Clairvoyant estimator closely. So it may be used as a sub-optimum estimator in this case. \\

\begin{figure}[!t]
\begin{centering}
\includegraphics[width=3.5 in]{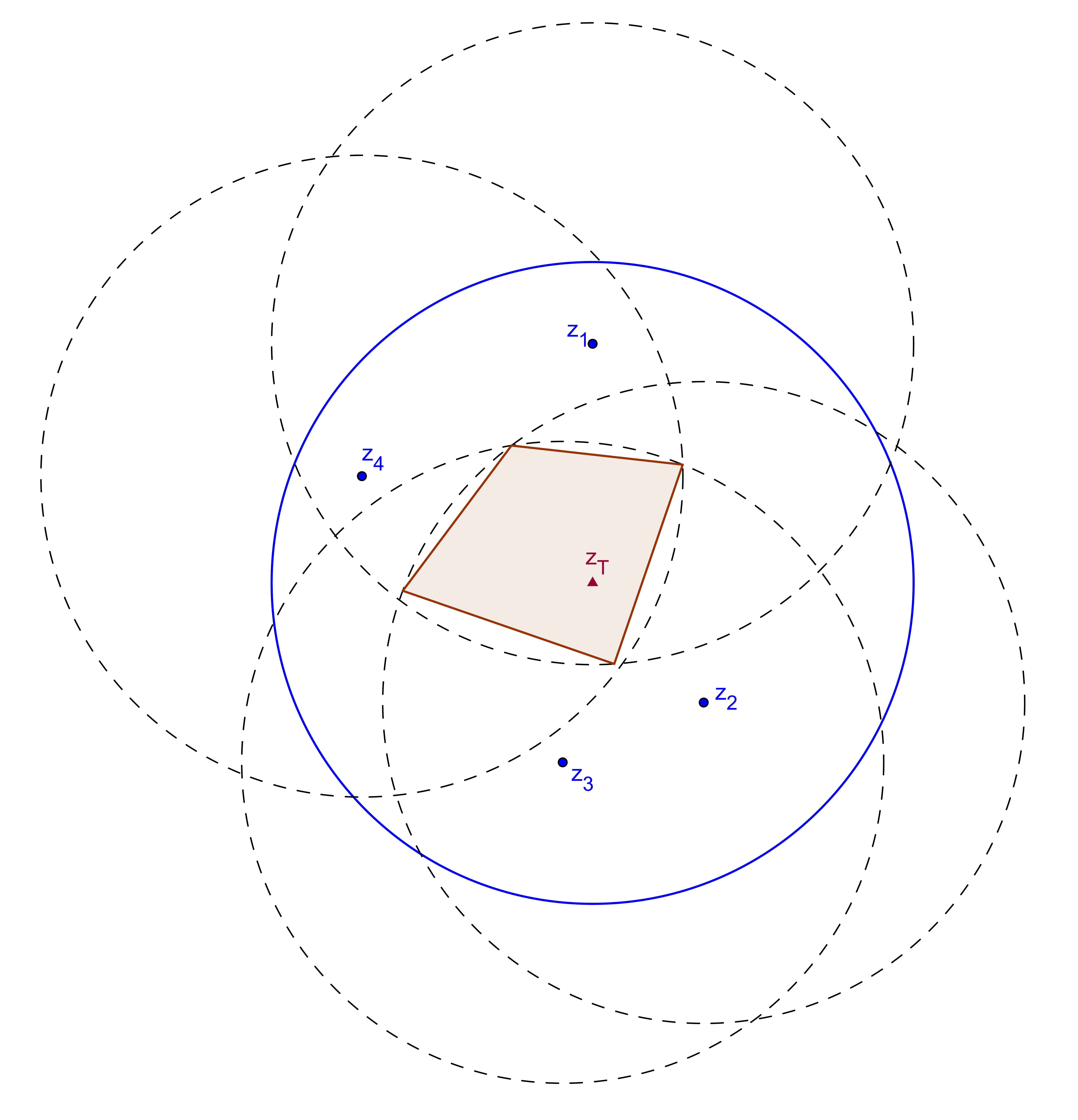}
\caption{An illustration of $\mathcal{SLT}(\mathbf{Z})$}
\label{IllustrationSuboptimal}
\end{centering}
\end{figure}

\section{Simulation Results}
\label{sec:SimulationResults}
A set of simulations are performed in two dimensional space to verify that for known $R$, the center of mass of $\mathcal{T}(\mathbf{Z})$ is the minimum variance estimator among other famous estimators. We assume that in each trial $N$ sensors are dispersed  randomly in a rectangular region $A = 100m\times 100m$, centered at the origin, where the target is located. Number of sensors depend on the density of sensor deployment, $\rho$ and is assumed to be $N=\lfloor \rho \times A \rfloor$. Any sensors located within the $R$ distance of the origin would be a detecting sensor and vice versa. In \cite{ArianSpawc2010}, a number of heuristic algorithm introduced for estimation of target location when only the location of detecting sensors will be reported. In these methods, the target location will be estimated as the Steiner center, the center of the minimum enclosing rectangle or the center of the minimum enclosing circle for the locations of detecting sensors denoted by Steiner center, center of MER, and center of MEC respectively.
For the purpose of comparison, the center of mass of \textit{possible target region} based on the detecting sensors, $CM\left(\mathcal{T}(\mathbf{Z})\right)$, have been considered along with the above mentioned heuristic methods. An extension of a grid base algorithm has been used to find the center of mass of \textit{possible target region}. In addition, Delaunay algorithm readily available in Matlab is employed to implement the triangulation used in sub-optimum estimator \cite{Shewchuk2002Delaunay,Lin1996Delaunay}.
Ignoring the trials that result in no detecting sensor, mean square error of these methods have been calculated for each case.

Figures \ref{OnlyD_R_1_Tr8000} and \ref{OnlyD_R_5_Tr8000} depict Mean Square Error (MSE) of the above mentioned estimators versus density when $R$ is fixed and is equal to 1 and 5 respectively. The results have been averaged over 8000 trials. As can be seen in the graphs, the center of mass of $\mathcal{T}(\mathbf{Z})$ beats heuristic estimators as density increases. When the density is small, the likelihood of trials with only one or two sensors detecting is high which in these cases, all methods have identical estimates. It is also noticeable that for large densities, the sub-optimum estimator follows the $CM\left(\mathcal{T}(\mathbf{Z})\right)$ closely.

Figure \ref{OnlyD_D_1_Tr2000} and \ref{OnlyD_D_4_Tr2000} depict MSE of the above mentioned estimators versus $R$ when density of sensor deployment, $\rho$, is fixed and is equal to 1 and 4 respectively. The results have been averaged over 2000 trials. As is clear from these graphs, the center of mass of $\mathcal{T}(\mathbf{Z})$ beats the heuristic estimators again as  $R$ increases.

\begin{figure}[!t]
\begin{centering}
\includegraphics[width=3.5 in]{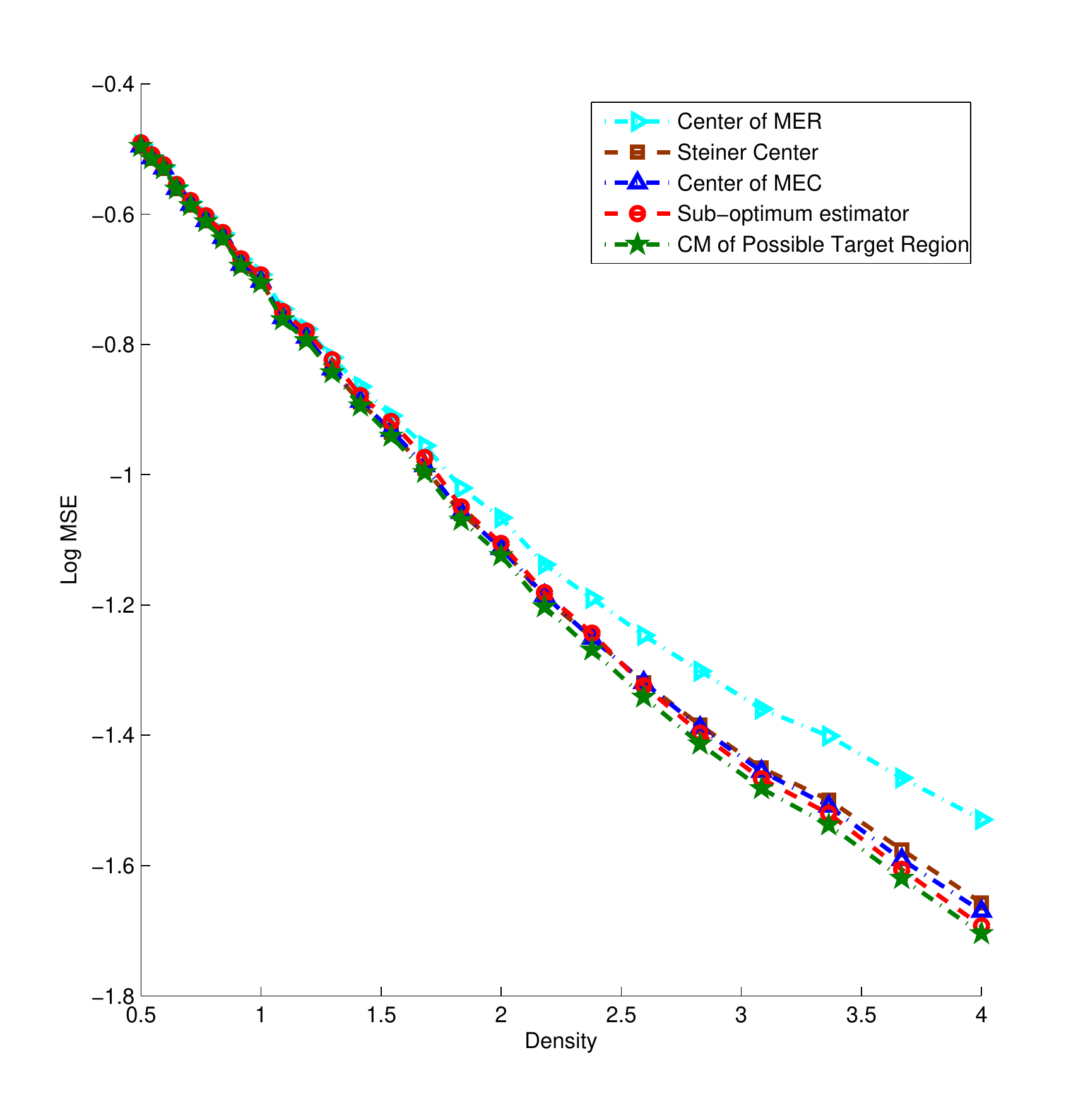}
\caption{MSE versus density for R=1}
\label{OnlyD_R_1_Tr8000}
\end{centering}
\end{figure}

\begin{figure}[!t]
\begin{centering}
\includegraphics[width=3.5 in]{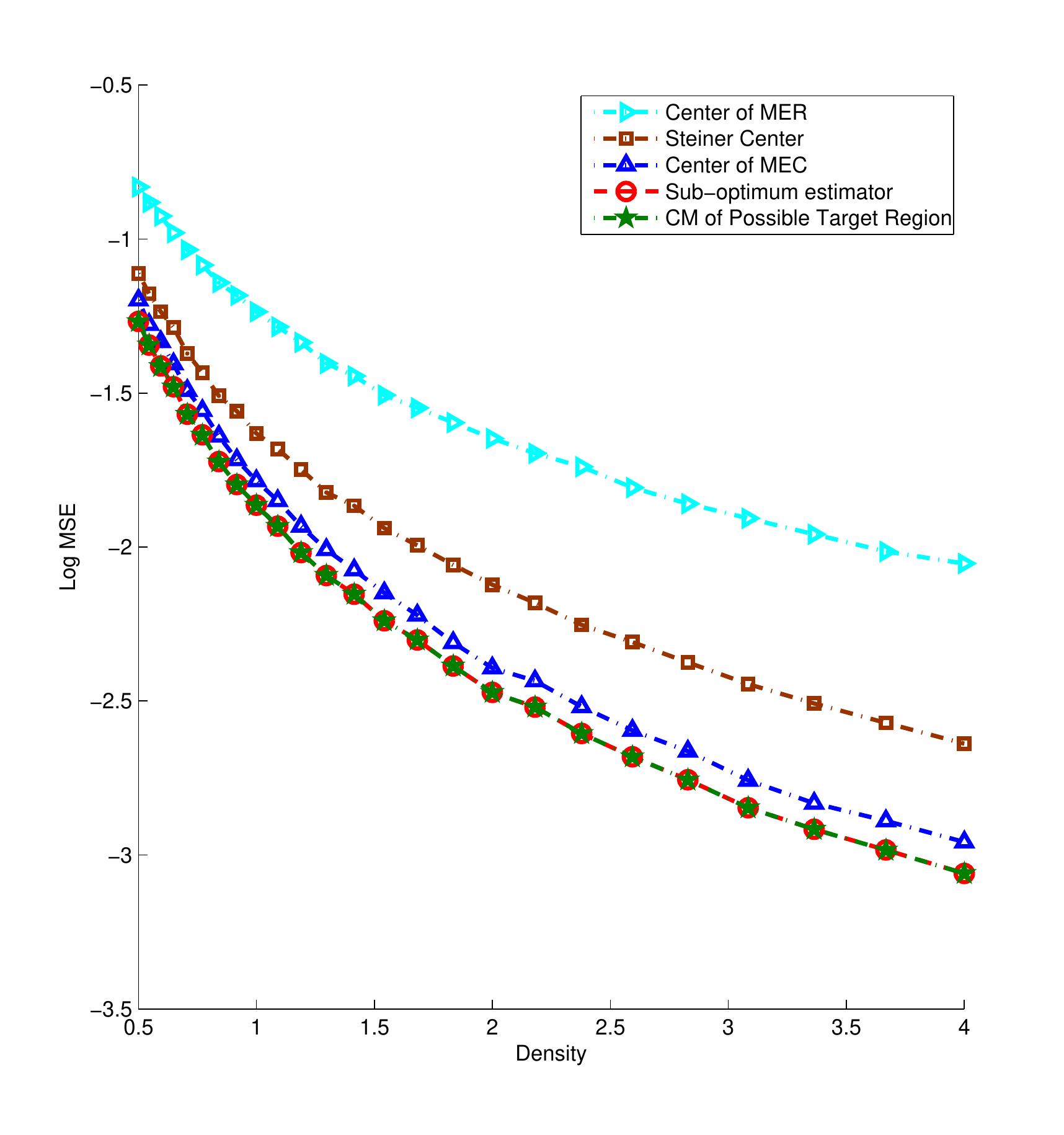}
\caption{MSE versus density for R=5 }
\label{OnlyD_R_5_Tr8000}
\end{centering}
\end{figure}

\begin{figure}[!t]
\begin{centering}
\includegraphics[width=3.5 in]{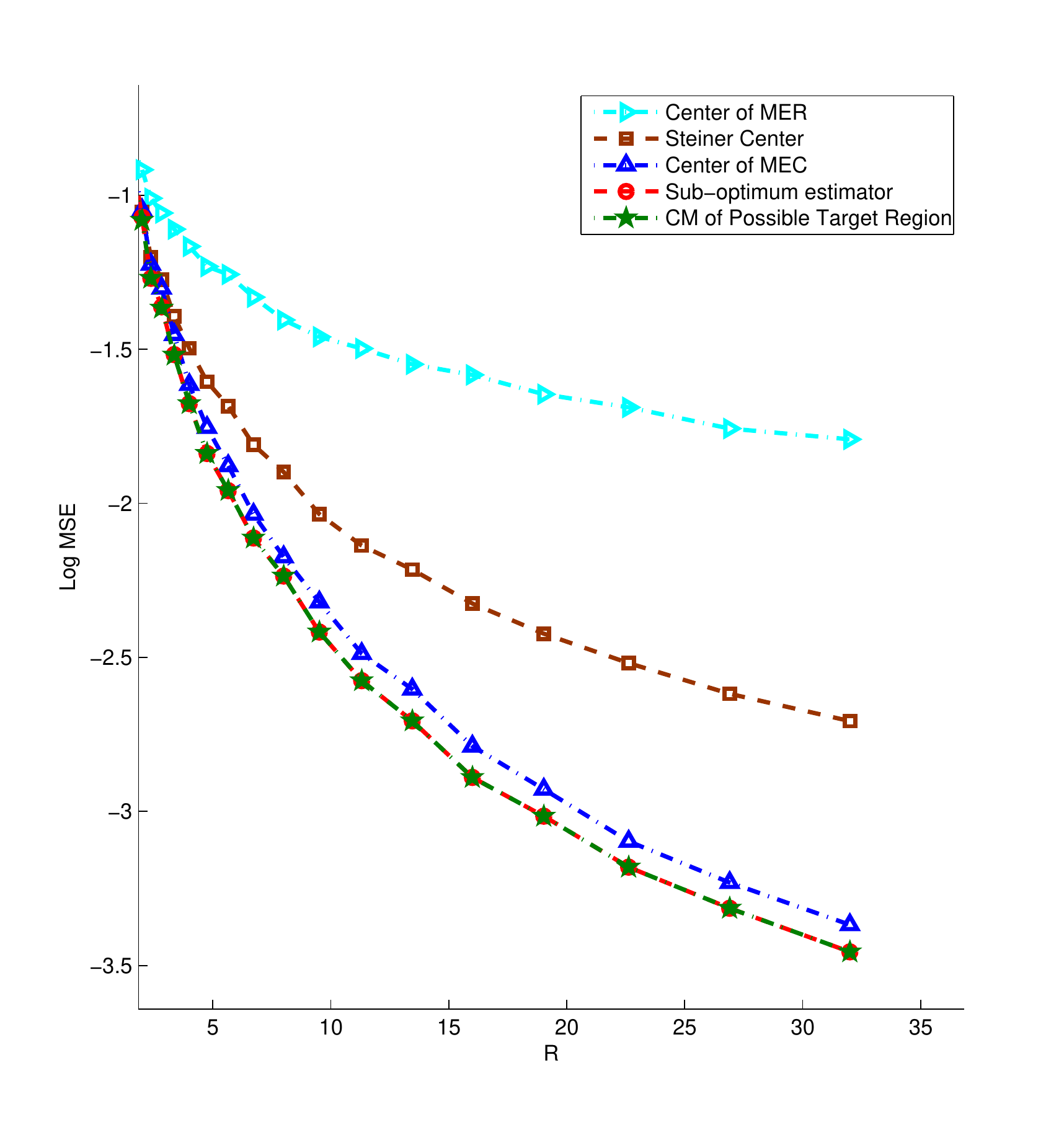}
\caption{MSE versus radius for $\rho=1$}
\label{OnlyD_D_1_Tr2000}
\end{centering}
\end{figure}

\begin{figure}[!t]
\begin{centering}
\includegraphics[width=3.5 in]{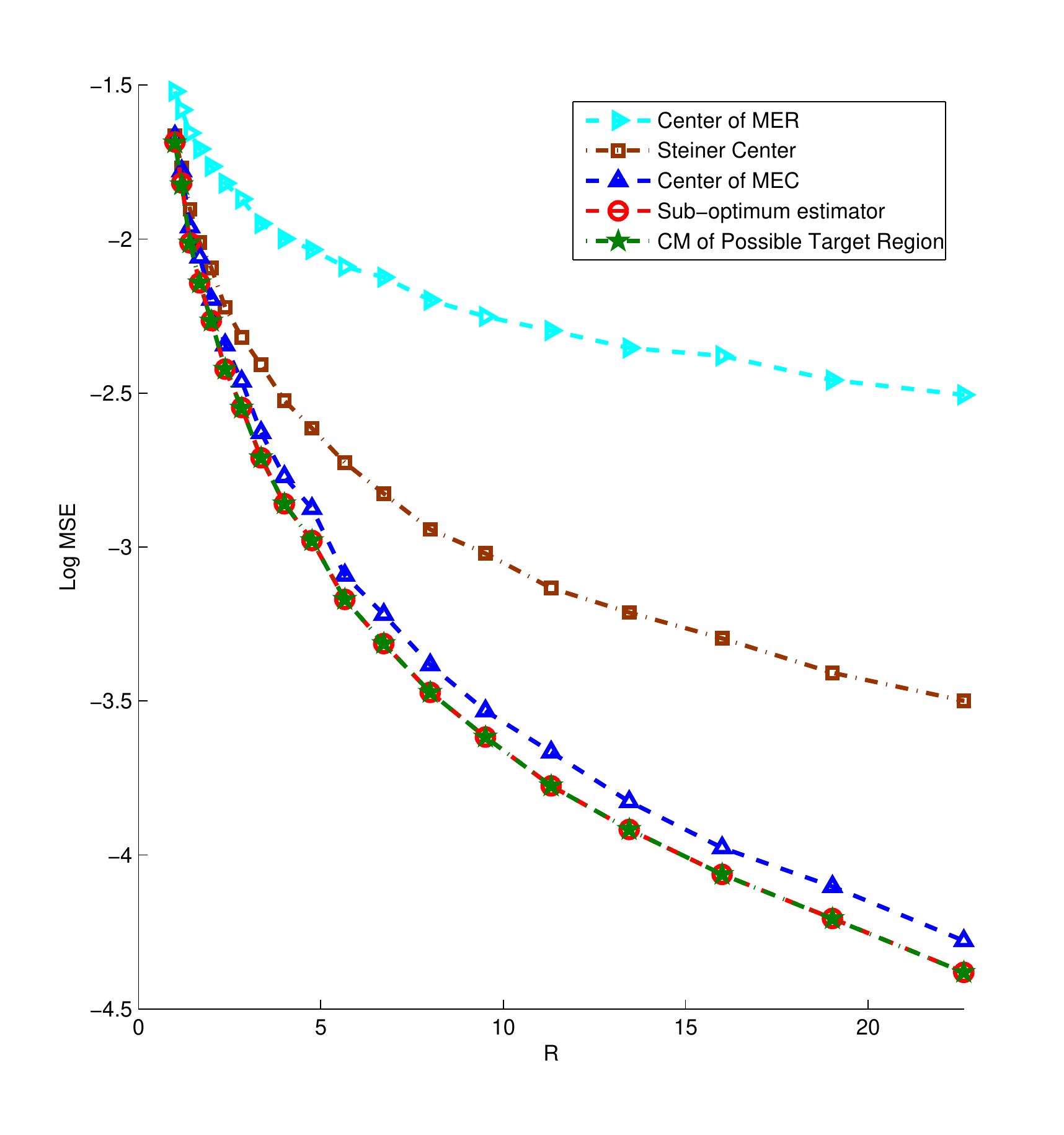}
\caption{MSE versus radius for $\rho=4$}
\label{OnlyD_D_4_Tr2000}
\end{centering}
\end{figure}

\section{Conclusion}
\label{sec:Conclusion}
This paper investigated the existence of an MVU location estimator when noise free detectors are deployed around a target and only the detecting sensors report their locations to a fusion center. It is proven mathematically that when the radius of detection is known and the target is located well inside the sensor deployment region, the center of mass of the \textit{possible target region} is the MVU estimator and when the radius of detection is not known the MVU estimator does not exist. Moreover, the minimal sufficient statistics of the detecting sensors are derived both when the radius of detection is known and when it is not known. In addition, a set of simulations is performed to compare the performance of the MVU estimator with various heuristic estimators. Finally, a sub-optimum estimators introduced, which is computationally less complex and the processing cost is independent of its resolution. It is also shown that when the density of deployment sensor is increased the performance of the sub-optimum estimator approaches the MVU performance.


\bibliographystyle{IEEEbib}
\bibliography{refs}
\end{document}